\documentclass[10pt]{article}
\usepackage{amsmath,amssymb,amsthm,times}
\usepackage[margin=1in]{geometry}
\usepackage{enumerate,multirow,tikz}
\definecolor{Wongblue}{HTML}{0072B2}
\definecolor{Wongorange}{HTML}{E69F00}
\usepackage{mathtools, cite}
\usepackage{float}
\usepackage{hyperref}
\usepackage{bbm}
\usepackage[noend]{algpseudocode}
\usepackage{algorithm}
\usepackage[shortlabels]{enumitem}
\mathtoolsset{showonlyrefs}

\newcommand{\indicator}[1]{\mathbbm{1}_{#1}}
\newcommand{\Vol}[0]{\operatorname{Vol}}
\newcommand{\Var}[0]{\operatorname{Var}}
\newcommand{\med}[0]{\operatorname{median}}
\newcommand{\R}[0]{\mathbb{R}}
\newcommand{\E}[0]{\mathbb{E}}

\newcommand{\OO}[0]{\mathcal{O}}
\newcommand{\FF}[0]{\mathcal{F}}
\DeclareMathOperator*{\argmin}{arg\,min}

\renewcommand{\complement}{\mathsf{c}}

\newtheorem{theorem}    		  {Theorem}
\newtheorem{lemma}      [theorem] {Lemma}

\theoremstyle{definition}
\newtheorem{example}    		  {Example}

\begin{document}
\title{Performance bounds for nearest neighbor search with $k$-d trees%
\thanks{The authors are grateful to the National Science Foundation for funding under grant CCF-2217058.}}
\author{Marco Bazzani\thanks{University of Washington,
\texttt{mbazzani@cs.washington.edu}.} \and Sanjoy Dasgupta\thanks{University of
California San Diego, \texttt{dasgupta@eng.ucsd.edu}.}}
\date{}
\maketitle

\begin{abstract}
The $k$-d tree is one of the oldest and most widely used data structures for
nearest neighbor search. It partitions Euclidean space into axis-aligned
rectangular cells. There are two standard ways to find the nearest neighbor to a query in a
$k$-d tree.
Defeatist search returns the closest data point in the query's cell, while
comprehensive search also searches other cells as needed to guarantee it finds
the nearest neighbor.
Both strategies are commonly believed to perform poorly in high
dimensions, but there have been few theoretical results explaining this.
We prove non-asymptotic bounds on the runtime of comprehensive search and
the accuracy of defeatist search.
Under mild distributional assumptions, when the dimension $d$ is at least
polylogarithmic in the number of data points, defeatist search is no more
likely to return the nearest neighbor than random guessing, and
comprehensive search visits every cell with high probability.
We also show that on uniform data, with high probability, comprehensive search
visits at most $2^{\OO(d)}$ cells when each cell contains at least
logarithmically many data points, and defeatist search returns the nearest
neighbor when each cell additionally contains at least $2^{\OO(d \log d)}$ data
points.
Finally, for arbitrary absolutely continuous distributions, we upper bound the
expected distance between the query and the point returned by defeatist search.
\end{abstract}

\section{Introduction}

The $k$-d tree, introduced by Bentley \cite{Bentley1975}, is one of the oldest
and most widely used data structures for nearest neighbor search \cite{Vempala2012, AryaMNSW1998, RamS2019, MujaL2014}. 
Given $n$ points in $\R^d$, a $k$-d tree is built by recursively partitioning space along
coordinate axes (Figure~\ref{fig:buildtree}). At each node, one coordinate is
selected and the data is split at the median value along
that coordinate, so that half of the data points are on each side of the split.
Splitting continues until each cell contains fewer than $2n_0$
points, where $n_0$ is a parameter called the minimum leaf size. 
In practice, $n_0$ is usually chosen to be a constant (for instance, the popular
scikit-learn implementation uses a default of 40~\cite{sklearnKDTree}) or growing
polylogarithmically in the number of data points.
Common strategies for choosing the split coordinate include round-robin cycling,
selecting the coordinate along which the data has maximum variance or range, 
and random selection. We study the original round-robin variant.
The result of this procedure can be viewed either as a tree data structure
(Figure~\ref{fig:buildtree}) or as a hierarchical partition of Euclidean space
into rectangular cells (Figure~\ref{fig:kdtree}).

\begin{figure}[H]
\begin{algorithm}[H]
	\begin{algorithmic}
	\Function{build\_tree}{$S, n_0, l$}
	\If{$|S| < 2 n_0$}
		\State \Return $S$
	\EndIf
	\State $d \gets \text{ambient dimension of } S$
	\State $\alpha \gets (l \bmod d) + 1$
	\State $s \gets \med\{p_\alpha : p \in S\}$
	\State $\text{rule}(x) \gets (x_\alpha \le s)$
	\State $\text{left} \gets \Call{build\_tree}{\{ p \in S : \text{rule}(p) = \mathrm{true} \}, n_0, l+1}$
	\State $\text{right} \gets \Call{build\_tree}{\{ p \in S : \text{rule}(p) = \mathrm{false} \}, n_0, l+1}$
	\State \Return $(\text{rule}, \text{left}, \text{right})$
	\EndFunction
	\end{algorithmic}
\end{algorithm}
\caption{Procedure for building a $k$-d tree with minimum leaf size $n_0$ from a set of
data points $S$. Invoked with $l = 0$.}
\label{fig:buildtree}
\end{figure}

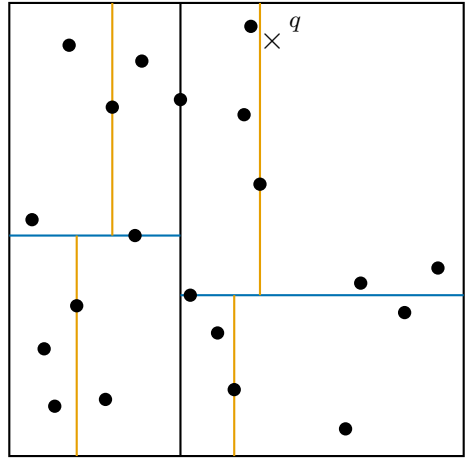
\begin{figure}[H]
\centering
\begin{tikzpicture}
  \draw[thick] (0,0) rectangle (6,6);

  \draw[black, thick] (2.26,0) -- (2.26,6);

  \draw[Wongblue, thick] (0,2.92) -- (2.26,2.92);
  \draw[Wongblue, thick] (2.26,2.13) -- (6,2.13);

  \draw[Wongorange, thick] (0.89,0) -- (0.89,2.92);
  \draw[Wongorange, thick] (1.36,2.92) -- (1.36,6);
  \draw[Wongorange, thick] (2.97,0) -- (2.97,2.13);
  \draw[Wongorange, thick] (3.31,2.13) -- (3.31,6);

  \foreach \x/\y in {
    0.3/3.13, 0.46/1.42, 0.6/0.66, 0.79/5.44,
    0.89/1.99, 1.27/0.75, 1.36/4.62, 1.66/2.92,
    1.75/5.23, 2.26/4.72, 2.39/2.13, 2.75/1.63,
    2.97/0.88, 3.1/4.52, 3.19/5.69, 3.31/3.6,
    4.44/0.36, 4.64/2.29, 5.22/1.9, 5.66/2.49
  } {
    \fill (\x,\y) circle (2.5pt);
  }

  \node[font=\large] at (3.47,5.5) {$\times$};
  \node[anchor=south west, font=\small] at (3.57,5.5) {$q$};
\end{tikzpicture}
\caption{A $k$-d tree partitioning of $[0,1]^2$ with leaf size $n_0 = 2$.
The dots are data points, the cross marks a query point $q$.
The nearest neighbor to $q$ lies in an adjacent cell.
The first split is black, the next level of splits is blue, and the last is orange.}
\label{fig:kdtree}
\end{figure}

There are two standard ways to use a $k$-d tree for nearest neighbor search. The
first is \emph{defeatist search} (Figure~\ref{fig:defeatist}), which finds
the cell that the query $q$ lands in and returns the closest point
there. This requires at most $2n_0$ distance computations, but the true nearest
neighbor to $q$ may lie in a different cell. For example, this occurs in Figure~\ref{fig:kdtree}. 
\emph{Comprehensive search} (Figure~\ref{fig:comprehensive}) 
addresses this by backtracking: whenever a
sibling cell intersects the ball centered at $q$ with radius equal to
the distance of the closest point found so far, that sibling is searched as well. This
guarantees that the nearest neighbor is found, but may require many more
distance computations.

\begin{figure}[H]
\begin{algorithm}[H]
	\begin{algorithmic}
	\Function{defeatist\_search}{node, $q$}
	\If{node is a leaf}
		\State \Return $\displaystyle \argmin_{p \in \text{points}(\text{node})} \|p - q\|$
	\EndIf
	\If{$\text{rule}(\text{node})(q)$}
		\State \Return \Call{defeatist\_search}{\text{left}(\text{node}), $q$}
	\Else
		\State \Return \Call{defeatist\_search}{\text{right}(\text{node}), $q$}
	\EndIf
	\EndFunction
	\end{algorithmic}
\end{algorithm}
\caption{Defeatist search procedure for the nearest neighbor to a query point $q$ in a
$k$-d tree. Invoked with node set to the root of the tree.}
\label{fig:defeatist}
\end{figure}

\begin{figure}[H]
\begin{algorithm}[H]
	\begin{algorithmic}
	\Function{comprehensive\_search}{node, $q$, best}
	\If{node is a leaf}
		\State \Return $\displaystyle
		\argmin_{p \in \text{points}(\text{node}) \cup \{\text{best}\}} \|p - q\|$
	\EndIf
	\If{$\text{rule}(\text{node})(q)$}
		\State $(\text{first}, \text{second}) \gets (\text{left}(\text{node}), \text{right}(\text{node}))$
	\Else
		\State $(\text{first}, \text{second}) \gets (\text{right}(\text{node}), \text{left}(\text{node}))$
	\EndIf
	\State $\text{best} \gets \Call{comprehensive\_search}{\text{first}, q, \text{best}}$
	\If{$\text{cell}(\text{second})$ intersects $B_{\|q - \text{best}\|}(q)$}
		\State $\text{best} \gets \Call{comprehensive\_search}{\text{second}, q, \text{best}}$
	\EndIf
	\State \Return best
	\EndFunction
	\end{algorithmic}
\end{algorithm}
\caption{Comprehensive search procedure for the nearest neighbor to a query point $q$
in a $k$-d tree. Invoked with node set to the root of the tree and
$\text{best} = \text{null}$.}
\label{fig:comprehensive}
\end{figure}

It has long been empirically observed that $k$-d trees
perform poorly on high dimensional data \cite{Sproull1991,
WeberSB1998, NeneN1997}. Despite this empirical evidence, the theoretical
picture has remained incomplete.
In this paper we prove non-asymptotic bounds on the performance of defeatist and
comprehensive search as a function of $d$ and $n$. 
Specifically, we prove two lower bounds for $k$-d trees constructed from $n$ data
points in $d$ dimensions with leaf size $n_0$. For absolutely
continuous product distributions on $[0,1]^d$ with marginal density bounded
below, defeatist search returns the nearest neighbor with probability at
most $\OO(n_0/n)$ once $d \in \Omega(\log^3 n)$.
(Theorem~\ref{thm:defeatistLowerBound}).
For absolutely continuous distributions on $[0,1]^d$ with bounded density,
comprehensive search visits every leaf cell with high probability for $d \in
\Omega(\log n)$ (Theorem~\ref{thm:comprehensiveLowerBound}).

We also prove complementary upper bounds for uniform data on $[0,1]^d$. Defeatist
search returns the nearest neighbor with high probability when
$n_0$ is at least $d^{\OO(d)}$ and $n_0 \in \Omega(\log n)$
(Theorem~\ref{thm:defeatistUpperBound}). Comprehensive search visits at most
$2^{\OO(d)}$ cells with high probability when $n_0 = \Omega(\log n)$
(Theorem~\ref{thm:comprehensiveUpperBound}). Our upper bounds apply only to
uniform data, which is a limitation. However, uniform data is not a setting
where $k$-d trees perform unusually well. Our lower bounds are strongest on
uniform data, and Example \ref{ex:goodperformance} in Section
\ref{sec:lowerbounds} where $k$-d trees perform well in high dimensions
is highly non-uniform. This suggests uniform data is in
fact the worst case for $k$-d tree performance, though it is also the easiest
to analyze. 

Finally, for general absolutely continuous distributions on $[0,1]^d$, we bound
the expected diameter of the leaf cell containing the query by
$\OO\left(d \left(\frac{n_0}{n}\right)^{1/d}\right)$
(Theorem~\ref{thm:smallDiam}). In particular, this bounds the distance between
the query and the point returned by defeatist search.

We are aware of two prior theoretical results on $k$-d tree performance.
Friedman, Bentley, and Finkel \cite{FriedmanBF1977} gave a heuristic argument
that for fixed $d$, the number of points examined by comprehensive search
remains bounded as $n \to \infty$, but the argument does not extend to growing
$d$. Panigrahy \cite{Panigrahy2008} considers a planted-query model on uniform
data in $[0,1]^d$. A random data point $p$ is selected, and a query point $q$
is planted that is about $c$ times closer to $p$ than to any other data point.
Specifically, $q$ is at distance about $r/c$ from $p$, where $r$ is the
distance from a random point in $[0,1]^d$ to its nearest data point. In this
setting, defeatist search returns the nearest neighbor with probability
at most $e^{-\Omega(d/c)}$. Panigrahy's analysis applies for $d \le \log n$,
disjoint from the $d \in \Omega(\log^3 n)$ of
Theorem~\ref{thm:defeatistLowerBound}. 
At $d=\log n$, Panigrahy's bound becomes
$n^{-\Omega(1/c)}$. In our regime, we obtain the bound $\OO(n_0/n)$.
Our result also holds for arbitrary $q$ and on a broader class of distributions.

Our lower bounds assume density bounded below or above on all of $[0,1]^d$. Some
similar assumption seems to be necessary. Example~\ref{ex:goodperformance} gives a
distribution on $[0,1]^d$ whose marginal densities vanish on most of $[0,1]$
and are large on the rest, for which both defeatist and comprehensive search
perform extremely well even when $d = \log(n/n_0)$.

More broadly, our distributional assumptions can be interpreted as a full-dimensionality
requirement. Since smooth manifolds in $\R^d$ have Lebesgue measure zero, a
distribution concentrated near such a manifold would have very high density near
the manifold and very low density elsewhere, placing it outside the scope of our
theorems. This is consistent with the empirical picture. Nene and Nayar
\cite{NeneN1997} find that on real-world data with low intrinsic dimension, 
$k$-d trees can remain effective for nearest neighbor search even when the ambient
dimension is large. In contrast, $k$-d trees tend to perform poorly on synthetic,
full-dimensional data \cite{Sproull1991, WeberSB1998}. Characterizing $k$-d tree
performance for distributions with low intrinsic dimension remains open.
Dasgupta and Freund \cite{DasguptaF2008} show that there exist data sets in
$\R^d$ with Assouad dimension $\OO(\log d)$ on which a $k$-d tree still requires
$d$ levels to halve the cell diameter. On the other hand, Vempala
\cite{Vempala2012} shows that $k$-d trees on randomly rotated data adapt to
doubling dimension, which intuitively suggests that $k$-d trees adapt to
doubling dimension in most non-pathological cases.

\section{Preliminaries}

We use the following definitions and notation.
The median of $n$ points is the $\lceil n/2 \rceil$-th order statistic.
The aspect ratio of a rectangle $\prod_{i=1}^d [a_i, b_i]$ is
$\frac{\max_i (b_i - a_i)}{\min_j (b_j - a_j)}$.
For $x \in \R^d$, let $B_r(x) = \{ y \in \R^d : \|x-y\| < r\}$ denote the open
ball of radius $r$ centered at $x$, and let $\log$ denote the base-$2$ logarithm.

Throughout, we will assume for simplicity that $\log(n/n_0)$ is a positive integer and that the
data is drawn from an absolutely continuous distribution. As a result, with
probability $1$, the resulting $k$-d tree has depth exactly $\log(n/n_0)$,
with $n 2^{-l}$ points at every node of depth $l$.

\section{Lower Bounds} \label{sec:lowerbounds}

The results below show that both defeatist and
comprehensive search perform poorly when the data dimension is at
least polylogarithmic in the number of data points. In particular this
applies in the typical regime where $n$ is polynomial in $d$.

\subsection{Comprehensive Search Lower Bound}

\begin{theorem} \label{thm:comprehensiveLowerBound}
Let $q \in [0,1]^d$.
Draw $n$ data points iid from any absolutely continuous
distribution on $[0,1]^d$ with density bounded above by $M$, and build
a $k$-d tree with minimum leaf size $n_0$. Comprehensive search on $q$
visits every leaf cell with probability at least
\begin{align}
  1 - n M \left(\frac{2 \pi e \log(n/n_0)}{d}\right)^{d/2}
  \qquad\text{whenever}\qquad
  d \ge \log(n/n_0).
\end{align}
\end{theorem}
\begin{proof}
Let $p_1, \dots, p_n$ denote the data points and let $r = \sqrt{\log(n/n_0)}$.
The tree has depth at most $\log(n/n_0)$, so since $d \ge \log(n/n_0)$,
each coordinate is split at most once.
For unsplit coordinates, every leaf cell spans all of $[0, 1]$, so $q$ is
already in the cell along these coordinates.
For split coordinates, both $q$ and the cell are contained in $[0, 1]$, so
the distance along each is at most $1$.
Since there are at most $\log(n/n_0)$ split coordinates,
$d(q, C)^2 \le \log(n/n_0) = r^2$ for every leaf cell $C$.
In particular, every leaf cell intersects $B_r(q)$.
If additionally $\min_i \|p_i - q\| > r$, then
$B_{\min_i \|p_i - q\|}(q)$ intersects every leaf cell and comprehensive
search must visit all of them.
Since $p_i$ has density bounded by $M$,
\begin{align}
P(\|p_i - q\| \le r)
\le M \cdot \Vol(B_r(0))
= M \cdot \frac{(\pi \log(n/n_0))^{d/2}}{\Gamma(d/2+1)}
\le M \left(\frac{2\pi e \log(n/n_0)}{d}\right)^{d/2}
\end{align}
where the last inequality uses $\Gamma(d/2+1) \ge \left(\tfrac{d}{2e}\right)^{d/2}$.
A union bound over $p_1, \dots, p_n$ completes the proof.
\end{proof}
\subsection{Defeatist Search Lower Bound}

\begin{theorem} \label{thm:defeatistLowerBound}
There exist $c_1, c_2 > 0$ such that the following holds.
Let $\mu$ be a product distribution on $[0,1]^d$ whose
marginals have densities bounded below by $m > 0$.
Let $q \in [0,1]^d$.
Draw $n$ data points iid from $\mu$, and build a $k$-d tree
with minimum leaf size $n_0$. Defeatist search on $q$ returns the
nearest neighbor with probability at most
\begin{align}
  c_2 \frac{n_0}{n}
  \qquad\text{whenever}\qquad
  d \ge \frac{c_1}{m}\log^3 n.
\end{align}
\end{theorem}
\begin{proof}
Let $U$ be the projection onto the unsplit coordinates and $S$ be the projection
onto the split coordinates.

Write $\mathcal U = (U(p_1), \ldots, U(p_n))$ for
the unsplit coordinates of all $n$ data points, and condition on
$\mathcal U$. Then the split coordinates are still iid. Let $C(q)$ be the cell
that $q$ lands in. By symmetry of the $p_i$ and the fact that exactly $n_0$
of the $p_i$ end up in $C(q)$,
\begin{align}
	P(p_i \in C(q) \mid \mathcal U) = n_0/n
\end{align}
for any $i \in [n]$.

Now we identify a small set of candidate indices, determined by $\mathcal U$
alone, that must contain the nearest neighbor index. Because at most
$\log n$ coordinates are split along any root-to-leaf path and each contributes
at most $1$ to squared distance, $\|S(p_i - q)\|^2 \le \log n$ for every $i$.
Hence if $j^* = \argmin_j \|p_j - q\|^2$,
$$\|U(p_{j^*} - q)\|^2 \le \|p_{j^*} - q\|^2 \le \|p_i - q\|^2 \le \|U(p_i - q)\|^2 + \log n$$
for every $i$, and in particular for $i = \argmin_i \|U(p_i - q)\|^2$. So
$j^* \in \FF$, where
$$\FF = \{i \in [n] : \|U(p_i - q)\|^2 \le \min_j \|U(p_j - q)\|^2 + \log n\}.$$
As a result,
\begin{align}
	& P(\text{defeatist search returns the nearest neighbor} \mid \mathcal U) \\
	& \le P\left(\bigcup_{i \in \FF} \{p_i \in C(q)\} \mid \mathcal U\right) \\
	& \le |\FF|\frac{n_0}{n}
\end{align}
Taking the expectation over $\mathcal U$ gives
\begin{align}
 P(\text{defeatist search returns the nearest neighbor})
 \le \E[|\FF|] \frac{n_0}{n}.
\end{align}

We will show that $\E[|\FF|] = \OO(1)$. Let $M = \min_{i \ge 2} \|U(p_i - q)\|^2$.
By exchangeability, $P(\|U(p_1 - q)\|^2 \le M) = 1/n$. The strategy is to show,
via a custom anticoncentration bound (Lemma~\ref{lem:cdfratio}), that adding
$\log n$ to $M$ inflates this probability by at most a constant factor.
Let $F$ be the CDF of $\|U(p_1 - q)\|^2$. Then,
\begin{align}
\E[|\FF|]
& = \sum_i P(p_i \in \FF)	\\
& = n P(\|U(p_1 - q)\|^2 \le M + \log n)	\\
& = n \E[F(M + \log n)]	\label{eq:040}\\
& = n \left(\E[F(M + \log n) \indicator{F(M) < 1/n^2}] +  \E[ F(M +
\log n) \indicator{F(M) \ge 1/n^2}]\right) \label{eq:034}\\
& \le 1 +  n \E[ F(M +
\log n) \indicator{F(M) \ge 1/n^2}] \label{eq:039}
\end{align}
\eqref{eq:040} uses the fact that $p_1$ is independent of $p_2, \ldots, p_n$ and
hence of $M$, while \eqref{eq:039} follows from $F(M+\log n) \le 1$ and $P(F(M) < 1/n^2) \le \frac{n-1}{n^2} \le \frac{1}{n}$.

It remains to bound the second term in \eqref{eq:039}. let $W = \E[\|U(p_1{-}q)\|^2] -
\|U(p_1{-}q)\|^2$, a sum of $d - \log(n/n_0)$ independent centered variables in
$[-1,1]$. A quick computation shows that $\Var((p_{i,j}{-}q_j)^2) \ge
\frac{m}{180}$, hence $v = \Var(W) \ge (d{-}\log(n/n_0))\frac{m}{180} \ge
d\frac{m}{360} \ge \frac{c_1}{360}\log^3 n$ by the hypothesis $d \ge c_1 m^{-1}
\log^3 n$. We apply Lemma~\ref{lem:cdfratio} to $W$ with $\delta = 2\log n$ and $a = \E[\|U(p_1{-}q)\|^2] - t$, noting that $F(t) = P(W \ge a)$. For
$c_1$ a sufficiently large absolute constant and every $t$ with $F(t) \ge
1/n^2$, the hypotheses of Lemma~\ref{lem:cdfratio} hold, since $\delta \ge 1$,
$F(t) \ge 1/n^2 \ge e^{-\delta}$, and $\delta/v^{1/3}$ falls below any fixed
positive constant once $c_1$ is large. Hence, by monotonicity of $F$,
\begin{align}
\frac{F(t + \log n)}{F(t)} \le \frac{F(t + 2\log n)}{F(t)} \le c_3 \label{eq:Fratio}
\end{align}
for an absolute constant $c_3 > 0$.

We may apply \eqref{eq:Fratio} to $t = M$ since we are restricting to the event $F(M) \ge 1/n^2$. Hence,
\begin{align}
\E[F(M + \log n) \indicator{F(M) \ge 1/n^2}]
\le c_3 \E[F(M) \indicator{F(M) \ge 1/n^2}]
\le c_3\E[F(M)] = c_3/n
\end{align}
Combining with \eqref{eq:039} completes the proof.
\end{proof}

\begin{lemma}\label{lem:cdfratio}
There exist $c_1, c_2 > 0$ such that the following holds.
Let $Y_1, \ldots, Y_m$ be independent with $\E[Y_j] = 0$ and $|Y_j| \le 1$, and
let $W = \sum_{j=1}^m Y_j$ with $v = \Var(W)$. For every $a, \delta \in \mathbb{R}$  with $1 \le \delta \le c_1 v^{1/3}$
and $P(W \ge a) \ge e^{-\delta}$,
\begin{align}
\frac{P(W \ge a - \delta)}{P(W \ge a)} \le c_2.
\end{align}
\end{lemma}
\begin{proof}
Let $\Phi$ and $\phi$ denote the cumulative distribution function and
probability density function of the standard normal distribution.
The strategy is to compare $W$ to a Gaussian. We use a moderate-deviation result (Proposition~4.5
of~\cite{ChenFangShao2013}) to obtain a multiplicative bound $\frac{P(W \ge
z\sqrt v)}{1-\Phi(z)} \in [1/2,
 3/2]$ for $z \le a/\sqrt{v}$, reducing the problem to bounding the Gaussian tail ratio 
$\frac{1-\Phi(\frac{a -
\delta}{\sqrt{v}})}{1-\Phi(\frac{a}{\sqrt{v}})}$.
The hypotheses     
$\delta \le c_1 v^{1/3}$ and $P(W \ge a) \ge e^{-\delta}$ together keep both
$a$ and $\delta$ small, so this ratio is at most an absolute constant.

By \cite[Proposition~4.5]{ChenFangShao2013} applied to $W/\sqrt{v}$, there is an absolute constant $c_3 > 0$ such that for every $z \in [0, \sqrt v]$,
\begin{align}
\left| \frac{P(W \ge z\sqrt v)}{1 - \Phi(z)} - 1 \right| \le c_3 (1 + z^3) \gamma e^{4 z^3 \gamma}, \label{eq:md}
\end{align}
where $\gamma = v^{-3/2} \sum_j \E |Y_j|^3 e^{z|Y_j|/\sqrt v}$.

We will show that the right-hand side of \eqref{eq:md} is at most $1/2$ for $z
\in \left[0, a_{\scriptscriptstyle +}/\sqrt v\right]$. By inspection of \eqref{eq:md}, it suffices to show that $\gamma$ and $z^3 \gamma$ are both at most an absolute constant times $c_1^{3/2}$.

First, we check that $a_{\scriptscriptstyle +}/\sqrt v \le \sqrt v$ so that this range lies within
the domain of \eqref{eq:md}. Suppose $a > 0$, otherwise this is trivial. By Bernstein's inequality, $P(W \ge a) \le
\exp\left(-\frac{a^2}{2v + 2a/3}\right)$. Since our hypothesis is that $P(W \ge
a) \ge e^{-\delta}$,
\begin{align}
a^2 \le 2\delta(v + a). \label{eq:036}
\end{align}
The hypotheses $\delta \ge 1$ and $\delta \le c_1 v^{1/3}$ give $v \ge
c_1^{-3}\delta^3 \ge \delta$ for $c_1 \le 1$. Solving \eqref{eq:036} then yields $a \le \delta +
\sqrt{\delta^2 + 2\delta v}$, and since $v \ge \delta$ this means $a \le 3\sqrt{\delta v}$. Hence,
\begin{align}
\frac{a_{\scriptscriptstyle +}}{\sqrt v}
\le 3\sqrt\delta
\le 3\sqrt{c_1} v^{1/6}
\le \sqrt v,\label{eq:abound}
\end{align}
where the last inequality uses $c_1 \le 1/9$ and $v \ge \delta \ge 1$. If $a \le 0$ then
\eqref{eq:abound} holds trivially.

Now we bound $\gamma$ and $z^3 \gamma$. Since $|Y_j| \le 1$ and $\sum_j \E Y_j^2 = v$ we have $\gamma \le \tfrac{1}{\sqrt v} e^{z/\sqrt v}$.
Together with \eqref{eq:abound} and $v \ge c_1^{-3}$, this means $\gamma \le \frac{e}{\sqrt{v}} \le e c_1^{3/2}$.
In addition, \eqref{eq:abound} gives $z^3 \le 27 c_1^{3/2} \sqrt{v}$, so
$z^3 \gamma \le 27 e c_1^{3/2}$.
Then for $c_1$ sufficiently small and $z \in [0, a_{\scriptscriptstyle +}/\sqrt{v}]$,
\begin{align}
\tfrac{1}{2}(1 - \Phi(z)) \le P(W \ge z\sqrt v) \le \tfrac{3}{2}(1 - \Phi(z)). \label{eq:mdtwosided}
\end{align}
If $z < 0$ then \eqref{eq:mdtwosided} gives $ \tfrac{1}{2}(1 - \Phi(0)) \le P(W
\ge 0) \le P(W \ge z \sqrt{v})$. Furthermore, $P(W \ge z
\sqrt{v}) \le 1 = 2(1-\Phi(0))$.
Combining this with \eqref{eq:mdtwosided} and using the abbreviations $x = \frac{a_{\scriptscriptstyle +}}{\sqrt v}$ and $y = \frac{(a -
\delta)_{\scriptscriptstyle +}}{\sqrt v}$,
\begin{align}
\frac{P(W \ge a - \delta)}{P(W \ge a)} \le 4 \frac{1 - \Phi(y)}{1 - \Phi(x)}.
\end{align}
It remains to upper bound $\frac{1 - \Phi(y)}{1 - \Phi(x)}$. Birnbaum's inequality \cite{Birnbaum1942} gives $1 - \Phi(u) \ge
\frac{\sqrt{4+u^2}-u}{2}\phi(u)$. Since $\sqrt{4+u^2}-u \ge
\frac{2}{u+1}$ for $u \ge 0$, we have $1 - \Phi(u) \ge \frac{\phi(u)}{u+1}$,
or equivalently $\frac{\phi(u)}{1 - \Phi(u)} \le u+1$. Then
\begin{align}
\ln \frac{1 - \Phi(y)}{1 - \Phi(x)}
& = -\int_y^x \frac{d}{du} \ln(1 - \Phi(u)) \, du \\
& = \int_y^x \frac{\phi(u)}{1 - \Phi(u)} \, du \\
& \le \int_y^x (u+1)\,du \\
& = (x - y)\left(\frac{x + y}{2} + 1\right) \\
& \le \frac{\delta}{\sqrt{v}}(3\sqrt\delta + 1) \label{eq:035} \\
& \le \frac{4 \delta^{3/2}}{\sqrt v}  \label{eq:037} \\
& \le 4 c_1^{3/2} \label{eq:038},
\end{align}
where \eqref{eq:035} uses $x - y \le \delta/\sqrt{v}$ and the first inequality
in \eqref{eq:abound}, \eqref{eq:037} uses $\delta
\ge 1$, and \eqref{eq:038} uses $\delta \le c_1 v^{1/3}$. Therefore $P(W \ge a - \delta)/P(W \ge a)
\le 4 e^{4 c_1^{3/2}}$.
\end{proof}

Something like the density bounds in Theorems~\ref{thm:defeatistLowerBound}
and~\ref{thm:comprehensiveLowerBound} seems to be necessary. The following
example exhibits a distribution on $[0,1]^d$ with whose density is sometimes
zero and sometimes very large. 
When data is sampled from this distribution with
$d = \log(n/n_0)$, then comprehensive search only searches one cell and defeatist search returns the
nearest neighbor with high probability.

\begin{example}\label{ex:goodperformance}
Let $\epsilon = \frac{1}{10\sqrt d}$, $m = \frac{1}{d \log d}$, and
$\eta = \frac{1}{d^{10}}$. Let $\mu$ be the product distribution on
$[0,1]^d$ whose marginals are uniform on each of the intervals
$[0,\epsilon]$, $[\tfrac{1}{2}-\eta, \tfrac{1}{2}+\eta]$, and
$[1-\epsilon, 1]$, with masses $\tfrac{1-m}{2}$, $m$, and $\tfrac{1-m}{2}$
respectively. 

Let $n_0 = d^3 \log^3 d$ and $n = n_0 \cdot 2^d$, so then $d = \log(n/n_0)$. 
Draw $n$ data points
and a query point $q$ iid from $\mu$, and build a $k$-d
tree with minimum leaf size $n_0$. We claim that with probability
$1 - o(1)$ as $d \to \infty$, defeatist search returns the true nearest
neighbor and comprehensive search visits only $q$'s leaf cell.

The tree has depth $\log(n/n_0) = d$, so round-robin cycling splits each
coordinate exactly once along any root-to-leaf path. Call a coordinate
of a point \emph{left}, \emph{center}, or \emph{right} according to
which marginal interval it lies in, and call a point a \emph{corner
point} if no coordinate is center. We claim that each of the following holds
with probability $1 - o(1)$:
\begin{enumerate}[(a)]
  \item every split lies in $[\tfrac{1}{2}-\eta, \tfrac{1}{2}+\eta]$,
  \item $q$ is a corner point,
  \item every one of the $2^d$ corners contains a data point.
\end{enumerate}
For (a), an internal node has $N \ge n_0$ points and splits at the median
of the relevant coordinate. The median lies outside
$[\tfrac{1}{2}-\eta,\tfrac{1}{2}+\eta]$ only if more than $N/2$ of the
samples fall on the same side of the center band. By Hoeffding's inequality the
probability of this is at most
$2\exp(-Nm^2/2) \le 2\exp(-d\log d/2)$. A union bound over the
$< n$ internal nodes gives total failure probability at most
$2n\exp(-d\log d / 2) = o(1)$. For (b), each coordinate of $q$ is
center with probability $m$, so
$P(q \text{ is a corner point}) \ge 1 - dm = 1 - 1/\log d$ by a union
bound. For (c), any corner has probability mass $((1-m)/2)^d$, so
\begin{align}
P(\text{a corner contains no data points})
\le \left(1 - \tfrac{(1-m)^d}{2^d}\right)^n
\le \exp(-n_0 (1-m)^d)
\le \exp(-n_0/2)
\end{align}
for $d$ large, since $(1-m)^d \ge 1 - dm \to 1$. A union bound over the
$2^d$ corners gives total failure probability at most
$2^d \exp(-n_0/2) = o(1)$.

Now suppose (a), (b), and (c) occur. Let $p^*$ be a data point in $q$'s corner.
Then each coordinate of $p^* - q$ has absolute value at most $\epsilon$, so
$\|p^* - q\| \le \epsilon \sqrt d = 1/10$. Each coordinate of $q$ is at
distance at least $\tfrac{1}{2} - \eta - \epsilon \ge 1/10$ from the
corresponding split (for $d \ge 2$), so $B_{1/10}(q) \subseteq C(q)$.
Hence the nearest neighbor lies in $C(q)$, defeatist search returns
it, and no other leaf cell intersects $B_{1/10}(q)$, so comprehensive
search visits only $C(q)$.
\end{example}

\section{Upper Bounds}

\subsection{Comprehensive Search Upper Bound}
\begin{theorem} \label{thm:comprehensiveUpperBound}
There exists $c > 0$ such that the following holds.
Let $q \in [0,1]^d$.
Draw $n$ data points iid from the uniform distribution on
$[0,1]^d$, and build a $k$-d tree with minimum leaf size $n_0 \ge 11$.
The number of cells visited by comprehensive search on $q$ is at most
\begin{align}
  4 \cdot (8\sqrt{2\pi e})^d
  \qquad\text{with probability at least}\qquad
  1 - 3n \exp(-c n_0).
\end{align}
\end{theorem}

\begin{proof}
	By Lemma~\ref{lem:cellRegularity}, with probability at least
	$1 - 3n\exp(-cn_0)$, every leaf cell has aspect ratio at most $4$ and
	volume within $\left[\frac{n_0}{2n}, \frac{2n_0}{n}\right]$.
	Suppose this event occurs.
	Then the diameter of every leaf cell is at most $\Delta = 4\sqrt{d}
	\left(\frac{2n_0}{n}\right)^{1/d}$ by Lemma
	\ref{lem:diameterBoundUsingAspectRatio}. The nearest neighbor
	to $q$ in its leaf cell is within distance $\Delta$, so every cell visited
	by comprehensive search must intersect $B_\Delta(q)$. Since all leaf cells
	have diameter at most $\Delta$, every visited cell is contained in
	$B_{2\Delta}(q)$. Since each cell has volume at
	least $\frac{n_0}{2n}$ and $\Gamma\left(\frac{d}{2} + 1\right) \ge
	\sqrt{\frac{d}{2e}}^d$, the number of visited cells is at most
	\begin{align}
		\frac{2n}{n_0} \Vol(B_{2\Delta}(q)) = \frac{2n}{n_0}
		\frac{\pi^{d/2}}{\Gamma\left(\frac{d}{2} + 1\right)} (2\Delta)^d =
		4 \frac{(8 \sqrt{\pi d})^{d}}{\Gamma\left(\frac{d}{2} + 1\right)}
		\le 4 \cdot (8\sqrt{2\pi e})^d.
	\end{align}
\end{proof}

The following Lemma says that with high probability every cell in the $k$-d tree
has small aspect ratio and volume close to $n_0/n$.
The idea is that each split is at the median of the data along some coordinate, and for uniform
data, the median concentrates around the midpoint. Specifically, the split at
level $l$ deviates from the midpoint by at most
$\epsilon_l = O(\sqrt{n_0 \cdot 2^l / n})$ with high probability. 
The volume and aspect ratio of the cell are controlled by $\prod_l (1 \pm
\epsilon_l)$, which remains bounded because
$\sum_l \epsilon_l$ converges. A union bound over all $\le 2 n$ internal nodes
then gives the failure probability.

\begin{lemma} \label{lem:cellRegularity}
Draw $n$ data points iid from the uniform distribution on $[0,1]^d$, and build
a $k$-d tree with minimum leaf size $n_0 \ge 11$.
The probability that every leaf cell has aspect ratio at most $4$
and probability mass within
$\left[\frac{n_0}{2n}, \frac{2n_0}{n}\right]$ is at least
\begin{align}
1 - 3 n \exp(- c n_0)
\end{align}
where $c = \left(\frac{(\sqrt{2}-1)\ln 2}{8}\right)^2$.
\end{lemma}
\begin{proof}
Let $u$ be a node in the tree with level $l$ and split
coordinate $\alpha$. Let $[a,b]$ be the cross section of $u$'s
cell
\footnote{Whether $a,b$ are included in the interval does not matter for our argument.}
along the $\alpha$-th coordinate. At least $2^{-l} n - 1$ of points in $u$ are not
previous split points along coordinate $\alpha$. Hence, these points are
uniformly distributed on $[a,b]$.
Now let $\epsilon_l = c_1 \sqrt{\tfrac{n_0}{n}2^l}$ with $c_1 =
\frac{(\sqrt{2}-1)\ln 2}{8}$ and let $A_u$ be the event that $u$'s split position is between
$a + (b{-}a)\left(\tfrac{1}{2}{\pm}\epsilon_l\right)$.
Let $m = 2^{-l}n - 1$.
Since $n_0 \ge 11$, we have $\epsilon_l > \frac{1}{2m}$, so
Lemma \ref{lem:uniformMedian} gives
\begin{align}
	P(A_u)
	& \ge 1 - 2 \exp\left(- 2 m \left(\epsilon_l - \frac{1}{2m}\right)^{2}\right)
	\ge 1 - 3 \exp\left( - c_1^2 n_0\right)  \label{eq:010}
\end{align}
where the last step uses $m \ge 2^{-l-1} n$, $\epsilon_l \le c_1/\sqrt{2}$, and $2e^{\sqrt{2} c_1} \le 3$.
Now let $A$ be the event that $A_u$ occurs for every internal node of the tree.
Using \eqref{eq:010} and a union bound,
\begin{align}
P(A)
& \ge 1 - 3n \exp( - c_1^2 n_0). \label{eq:014}
\end{align}
When $A$ occurs, the aspect ratio of every cell is at most $2
\prod_{l=0}^{\log(n/n_0)-1} \frac{1/2 + \epsilon_l}{1/2 - \epsilon_l}$.
Furthermore,
\begin{align}
\log\left(2 \prod_{l=0}^{\log(n/n_0)-1} \frac{1/2 + \epsilon_l}{1/2 - \epsilon_l}\right)
& = 1 + \sum_{l=0}^{\log(n/n_0)-1} \log\left(\frac{1/2 + \epsilon_l}{1/2 - \epsilon_l}\right)\\
& \le 1 + \sum_{l=0}^{\log(n/n_0)-1} \log\left(1 + 8 \epsilon_l \right) \label{eq:011}\\
& \le 1 + \frac{8}{\ln 2} \sum_{l=0}^{\log(n/n_0)-1} \epsilon_l  \label{eq:012}\\
& = 1 + \frac{8c_1}{(\sqrt{2}-1)\ln 2}\left(1-\sqrt{\tfrac{n_0}{n}}\right)  \label{eq:013} \\
& \le 1 + \frac{8c_1}{(\sqrt{2}-1)\ln 2} = 2,
\end{align}
where \eqref{eq:011} follows from $\epsilon_l \le c_1 \le 1/4$ and $\frac{1/2 + x}{1/2 -
x} \le 1+8x$ for $0 \le x \le 1/4$; \eqref{eq:012} follows from $\log(1+x) \le
\frac{x}{\ln 2}$ for $x \ge 0$; \eqref{eq:013} follows from reindexing by
$\log(n/n_0) - l$ and summing the resulting geometric series.
Thus, every cell has aspect ratio at most $2^2 = 4$ when $A$ occurs.

Furthermore, the probability mass of every leaf cell is at least
$\prod_{l=0}^{\log(n/n_0)-1} (1/2 - \epsilon_l)$.
Since $\sum_{l=0}^{\log(n/n_0)-1} \epsilon_l \le \frac{c_1}{\sqrt{2}-1}$
by the geometric series computation above,
\begin{align}
\ln\left( \prod_{l=0}^{\log(n/n_0)-1} (1/2 - \epsilon_l) \right)
& = -\ln(n/n_0) + \sum_{l=0}^{\log(n/n_0)-1} \ln\left(1 - 2\epsilon_l \right) \\
& \ge -\ln(n/n_0) - \sum_{l=0}^{\log(n/n_0)-1} 4 \epsilon_l  \label{eq:019}\\
& \ge -\ln(n/n_0) - \frac{4c_1}{\sqrt{2}-1} = -\ln(n/n_0) - \frac{\ln 2}{2},
\end{align}
where \eqref{eq:019} follows from $\ln(1-x) \ge -2x$ for $0 \le x \le 1/2$.
Thus, every cell has probability mass at least $\frac{n_0}{n} e^{-\ln 2/2} =
\frac{n_0}{n\sqrt{2}} \ge \frac{n_0}{2n}$ when $A$ occurs.
A similar computation with $\prod_{l=0}^{\log(n/n_0)-1} (1/2 + \epsilon_l)$
shows that every cell has probability mass at most $\frac{\sqrt{2}n_0}{n} \le \frac{2n_0}{n}$ when $A$
occurs, which together with \eqref{eq:014} completes the proof.
\end{proof}
\subsection{Defeatist Search Upper Bound}
\begin{theorem} \label{thm:defeatistUpperBound}
There exists $c > 0$ such that the following holds.
Draw $n$ data points and a query $q$ iid from the uniform distribution on
$[0,1]^d$, and build a $k$-d tree with minimum leaf size $n_0 \ge 11$.
Defeatist search on $q$ returns the nearest neighbor with probability at least
\begin{align}
  1 - 3 n \exp(-c n_0) - 24 d^{3/2}\left( \frac{\log n_0}{n_0} \right)^{1/d}
  \qquad\text{whenever}\qquad n_0 \ge 2\log n. \label{eq:017}
\end{align}
\end{theorem}
\begin{proof}
Let $C$ be the leaf cell $q$ lands in.
Let $F$ be the event that $C$ has aspect ratio at most $4$.
Let $G$ be the event that $B_{\|q - \hat p\|}(q) \subset C$, where $\hat p$ is the
nearest neighbor to $q$ in $C$.
Whenever $G$ occurs, defeatist search returns the nearest neighbor to $q$.
Hence, it suffices to show that $P(G\cap F)$ is at least the quantity in \eqref{eq:017}.
Condition on the tree $T$. Then $F$ is determined by $T$, and on $\{q \in C\}$ the
query $q$ together with the $n_0 - \log(n/n_0) \ge n_0/2$ non-split-point data points
in $C$ are iid uniform on $C$. Then Lemma \ref{lem:defeatistUpperBoundEventG} yields
\begin{align}
	P(G \mid F) 
	\ge 1 - 12 d^{3/2}\left( \frac{\log(n_0/2)}{n_0/2} \right)^{1/d}
	\ge 1 - 24 d^{3/2}\left( \frac{\log n_0}{n_0} \right)^{1/d}. \label{eq:016}
\end{align}
Furthermore, Lemma \ref{lem:cellRegularity} gives
\begin{align}
	P(F) \ge 1 - 3 n \exp(- c n_0). \label{eq:015}
\end{align}
where $c = \left(\frac{(\sqrt{2}-1)\ln 2}{8}\right)^2$.
Combining \eqref{eq:016} and \eqref{eq:015} completes the proof.
\end{proof}

\begin{lemma} \label{lem:defeatistUpperBoundEventG}
Let $q$ and $m \ge 3$ data points be drawn iid from the uniform distribution on a
$d$-dimensional rectangle $R$ with aspect ratio $\alpha$. Then
\begin{align}
P(B_{\min_i \|q-p_i\|}(q) \subset R)
\ge 1 -  3\alpha d^{3/2}\left( \frac{\ln m}{m} \right)^{1/d}
\end{align}
\end{lemma}

\begin{proof}
 Let $L_1, \dots, L_d$ be the
side lengths of $R$ and let $R_{-t} = \{x \in R: d(x, \partial R) \ge t\}$. 
For $t \in [0, \min_i L_i/2]$ we have $\Vol(R_{-t}) = \prod_{j=1}^d (L_j - 2t)$.
Note that
\begin{align}
P(B_{\min_i \|q-p_i\|}(q) \subset R)
    & \ge P( \min_i \|q - p_i\| \le t \mid q \in R_{-t})P(q \in R_{-t})
	\label{eq:005}
\end{align}
Let $v_d = \Vol(B_1(0))$ be the volume of the unit ball in $\R^d$.
Since $B_t(q) \subset R$ whenever $q \in R_{-t}$, and samples are uniform on $R$,
\begin{align}
P(\min_i \|q - p_i\| \le t \mid q \in R_{-t}) 
\ge 1 - (1 - \Vol(R)^{-1} v_d t^d)^m
\ge 1 - \exp(- \Vol(R)^{-1} m v_d t^d).
\end{align}
and 
\begin{align}
P(q \in R_{-t}) 
= \prod_{j=1}^d (1 - 2t/L_j) 
\ge 1 - \sum_{j=1}^d 2t/L_j
\ge 1 - 2dt/L_{\min}
\end{align}
where $L_{\min} = \min_i L_i$.
Choosing $t = \alpha d^{1/2} L_{\min} \left( \frac{\ln m}{m} \right)^{1/d}$ gives
\begin{align}
    P(q \in R_{-t}) & \ge 1 - 2\alpha d^{3/2}\left( \frac{\ln m}{m} \right)^{1/d}
	\label{eq:004}
\end{align}
and
\begin{align}
    P( \min_i \|q - p_i\| \le t \mid q \in R_{-t})
    & \ge 1 - \exp\left(- \Vol(R)^{-1} m v_d t^d \right) \\
    & \ge 1 - \exp\left(- (\alpha L_{\min} d^{1/2})^{-d} m  t^d \right)
	\label{eq:006} \\
    & = 1 - \exp(-\ln m)  \label{eq:007}\\
	& = 1 - \frac{1}{m}  \label{eq:008}
\end{align}
where \eqref{eq:006} follows from $\Vol(R) \le \alpha^d L_{\min}^d$ and $v_d \ge
d^{-d/2}$, and \eqref{eq:007} follows from the definition of $t$.
Finally, combining \eqref{eq:005}, \eqref{eq:004} and \eqref{eq:008} yields
\begin{align}
P(B_{\min_i \|q-p_i\|}(q) \subset R) 
& \ge \left(1 - \frac{1}{m}\right)\left(1 - 2\alpha d^{3/2}\left( \frac{\ln m}{m}
\right)^{1/d}\right) \\
& \ge 1 - \frac{1}{m} - 2\alpha d^{3/2}\left( \frac{\ln m}{m} \right)^{1/d} \\
& \ge 1 -  3\alpha d^{3/2}\left( \frac{\ln m}{m} \right)^{1/d},
\end{align}
where the last line follows from $\frac{1}{m} \le \left(\tfrac{\ln m}{m}\right)^{1/d}$
(this holds for $m \ge 3$ since $\ln m \ge 1 \ge m^{-(d-1)}$)
and $\alpha d^{3/2} \ge 1$.
\end{proof}

\subsection{Cell Diameter Upper Bound}

The following result gives an upper bound on the expected diameter of the leaf
cell containing the query $q$. In particular, it gives an upper bound on the distance
between $q$ and the point returned by defeatist search.
One notable property of this result is that it holds for any absolutely continuous
distribution on $[0,1]^d$. 

\begin{theorem} \label{thm:smallDiam}
There exists $c > 0$ such that the following holds.
Draw $n$ data points and a query $q$ iid from any absolutely continuous
distribution on $[0,1]^d$, and build a $k$-d tree with minimum leaf size $n_0$.
The expected $L^1$ diameter of the leaf cell containing $q$ is at most
\begin{align}
  6 d \left(\frac{n_0}{n}\right)^{1/d}
  \qquad\text{whenever}\qquad n_0 \ge c \log n.
\end{align}
\end{theorem}

\begin{proof}
Let $\mu$ be the distribution of the data and the query.
Let $\Delta(C)$ be the $L^1$ diameter of a cell $C$ and $\Delta_i(C)$
be its $i$-th side length.
For each level $l \in \{0, 1, \ldots, \log(n/n_0)\}$ define 
\begin{align}
A(l) = \sum_{C \text{ at level } l} \mu(C) \Delta(C), \qquad
A_i(l) = \sum_{C \text{ at level } l} \mu(C) \Delta_i(C).
\end{align}
Given the $k$-d tree $T$, $A(l)$ is known but $q$ is still
distributed according to $\mu$. Since $T$ has depth
$\log(n/n_0)$, this means that $\E[\Delta(C(q)) \mid T] = A(\log(n/n_0))$,
and therefore
\begin{align}
\E[\Delta(C(q))] = \E[A(\log(n/n_0))].
\end{align}

Now we upper bound $\E[A(\log(n/n_0))]$. Let $c' = \frac{(\sqrt{2}-1)\ln 2}{2}$ and $\delta_l = c' \sqrt{\frac{n_0}{n 2^{-l}}}$.
For any non-leaf cell $C$, write $C_L$ and $C_R$ for its left and right children.
Let $F$ be the event that $\left|\frac{\mu(C_L)}{\mu(C)} - \frac{1}{2}\right| \le
\delta_l$ for every non-leaf cell $C$ at level $l$.

Let $j = (l \bmod d) + 1$ be the split coordinate at level $l$. For any
non-leaf cell $C$ at level $l$ with children $C_L, C_R$, we have $\Delta_j(C_L)
+ \Delta_j(C_R) = \Delta_j(C)$. Now suppose $F$ occurs. Then, 
\begin{align}
\mu(C_L) \Delta_j(C_L) + \mu(C_R) \Delta_j(C_R)
& \le \left(\tfrac{1}{2} + \delta_l\right) \mu(C) \bigl(\Delta_j(C_L) + \Delta_j(C_R)\bigr) \\
& = \left(\tfrac{1}{2} + \delta_l\right) \mu(C) \Delta_j(C). \label{eq:084}
\end{align}
As a result, $A_j(l+1) \le \left(\tfrac{1}{2} + \delta_l\right) A_j(l)$.
For $i \ne j$ we instead have $A_i(l+1) = A_i(l)$.
Combining these,
\begin{align}
A(l+1) = \sum_{i=1}^d A_i(l+1)
\le \left(\tfrac{1}{2} + \delta_l\right) A_j(l) + \sum_{i \ne j} A_i(l)
\label{eq:041}
\end{align}
After $d$ levels each coordinate has been split exactly once.
Thus, applying \eqref{eq:041} $d$ times
and using $\delta_{l+k} \le \delta_{l+d-1}$ for $k \le d-1$ yields
\begin{align}
A(l+d) \le \left(\tfrac{1}{2} + \delta_{l+d-1}\right) A(l) \label{eq:031}
\end{align}
Using $A(0) = d$ and repeatedly applying \eqref{eq:031},
\begin{align}
A(\log(n/n_0))
\le d \prod_{\substack{d \mid l \\ l+d \le \log(n/n_0)}}
\left(\tfrac{1}{2} + \delta_{l+d-1}\right) \label{eq:028}
\end{align}
whenever $F$ occurs.
Now note that
\begin{align}
\log\left(\prod_{\substack{d\mid l \\ l+d \le \log(n/n_0)}} (1/2 + \delta_{l+d-1})\right)
& = \sum_{\substack{d\mid l \\ l+d \le \log(n/n_0)}} \log(1/2 + \delta_{l+d-1}) \\
& \le \sum_{\substack{d\mid l \\ l+d \le \log(n/n_0)}} \left(-1 + \frac{2\delta_{l+d-1}}{\ln 2}\right) \label{eq:030} \\
& \le -\left\lfloor \frac{\log(n/n_0)}{d} \right\rfloor + \frac{2}{\ln 2}\sum_{i=0}^{\log(n/n_0)-1} \delta_i \label{eq:029}
\end{align}
where \eqref{eq:030} follows from $\log(1/2 + x) \le -1 + \frac{2x}{\ln 2}$
for $x \ge 0$.
Reindexing by $j = \log(n/n_0) - i$, we have $\sum_{i=0}^{\log(n/n_0) - 1}
\delta_i =  c' \sum_{j=1}^{\log(n/n_0)} \sqrt{2}^{-j} \le
\frac{c'}{\sqrt{2}-1}$. Hence,
\begin{align}
	\E[\Delta(C(q))]
	& = \E[A(\log(n/n_0))] \\
	& = \E[A(\log(n/n_0)) \indicator{F}] + \E[A(\log(n/n_0)) \indicator{F^\complement}] \\
	& \le \E[A(\log(n/n_0)) \indicator{F}] + d \, P(F^\complement) \label{eq:042} \\
	& \le d \cdot 2^{-\left\lfloor \frac{\log(n/n_0)}{d}
	\right\rfloor + \frac{2 c'}{(\sqrt{2}-1)\ln 2}} + 2 d n \exp\left(-\tfrac{c'^2}{4} n_0\right) \label{eq:043} \\
	& \le 4 d \left(\frac{n_0}{n}\right)^{1/d} + 2 d n \exp\left(-\tfrac{c'^2}{4} n_0\right) \label{eq:032} \\
	& \le 6 d \left(\frac{n_0}{n}\right)^{1/d}, \label{eq:033b}
\end{align}
where \eqref{eq:042} uses $A(\log(n/n_0)) \le d$;
\eqref{eq:043} follows from \eqref{eq:028}, \eqref{eq:029}, and Lemma~\ref{lem:approxUpperBoundEventF};
\eqref{eq:032} uses $\frac{2c'}{(\sqrt{2}-1)\ln 2} = 1$ and $\lfloor x \rfloor
\ge x - 1$;
and \eqref{eq:033b} follows from $c \ge 2\ln 2 \cdot (2/c')^2$, which gives
$n \exp(-\tfrac{c'^2}{4} n_0) \le n^{-1} \le (n_0/n)^{1/d}$.
\end{proof}
\begin{lemma} \label{lem:approxUpperBoundEventF}
Draw $n$ data points iid from an absolutely continuous distribution,
and build a $k$-d tree with minimum leaf size $n_0 \ge \frac{4}{c} \log n$ for
any constant $c \in (0, 1]$.
Let $\delta_l = c \sqrt{\frac{n_0}{n 2^{-l}}}$.
The probability that $\left|\frac{\mu(C_L)}{\mu(C)} - \frac{1}{2}\right| \le
\delta_l$ for every non-leaf cell $C$ at level $l$ with left child $C_L$ is
at least
\begin{align}
1 - 2n \exp\left( - \frac{c^2}{4} n_0 \right).
\end{align}
\end{lemma}

\begin{proof}
Consider a node $u$ at depth $l$, which contains $n2^{-l}$ points. At most $l$ of
these were split points for ancestor nodes. The remaining points are
independent draws from $\mu$ restricted to the cell $C$ of $u$. 
By the probability integral
transform, $\frac{\mu(C_L)}{\mu(C)}$ is distributed as the
median of $n2^{-l}$ points in $[0,1]$, where at least $n2^{-l}-l$ are independent uniform random
variables.
Combining this with Lemma \ref{lem:uniformMedian},
\begin{align}
	P\left(\left|\tfrac{\mu(C_L)}{\mu(C)} - 1/2\right| > \delta_l\right)
	\le 2 \exp\left( -2 (n2^{-l} - l) \left(\delta_l - \frac{l}{2(n2^{-l}-l)}\right)^2\right).
\end{align}
Because $n_0 \ge \frac{2}{c}\log n$, we have $\delta_l \ge c
\sqrt{\frac{n_0}{n 2^{-l}}} \ge \frac{2\log n}{n 2^{-l}}$.
Furthermore, $2(n2^{-l} - l) \ge n2^{-l}$ since $n_0 \ge \frac{2}{c} \log n \ge 2\log n$ (using $c \le 1$).
As a result, $\frac{l}{2(n2^{-l}-l)} \le \frac{2\log n}{n 2^{-l}} \le \delta_l$, which
 means $\frac{l}{2(n2^{-l}-l)} \le \frac{\delta_l}{2}$ when $n_0 \ge \frac{4}{c}\log n$.
Combining this with Lemma \ref{lem:uniformMedian},
\begin{align}
	P\left(\left|\tfrac{\mu(C_L)}{\mu(C)} - \tfrac{1}{2}\right| > \delta_l\right)
	\le 2 \exp\left(-\tfrac{1}{4} n2^{-l} \delta_l^2\right)
	= 2 \exp\left(-\frac{c^2}{4} n_0 \right) .
\end{align}
Taking a union bound over $< n$ internal nodes completes the proof.
\end{proof}

\section{Technical Lemmas}

\begin{lemma} \label{lem:diameterBoundUsingAspectRatio}
Let $R$ be a rectangle in $\R^d$ with volume $V$, aspect ratio $\alpha$, and $L^2$ diameter $\Delta$.
Then,
$$\Delta \le \alpha \sqrt{d} V^{1/d} .$$
\end{lemma}
\begin{proof}
Let $s_1, \dots, s_d$ be the side lengths of $R$.
Let $s_{\max} = \max_i s_i$ and $s_{\min} = \min_i s_i = s_{\max}/\alpha$. Then,
$$ V = \prod_{i=1}^d s_i \ge s_{\max} \cdot (s_{\min})^{d-1} = s_{\max} \cdot
\left(\frac{s_{\max}}{\alpha}\right)^{d-1} = \frac{s_{\max}^d}{\alpha^{d-1}}. $$
Solving for $s_{\max}$ gives
$$ s_{\max} \le \left( V \alpha^{d-1} \right)^{1/d} 
\le  \alpha V^{1/d}.$$
As a result,
$$ \Delta \le \sqrt{d} s_{\max} \le \alpha \sqrt{d} V^{1/d} $$
\end{proof}

\begin{lemma} \label{lem:uniformMedian}
  Let $M_{n,k}$ be the median of $n$ points, where $k$ are fixed and $n-k$ are independent uniform random variables on $[0, t]$.
  If $\frac{\delta}{t} > \frac{k}{2(n-k)}$, then
  $$ P(|M_{n,k} - t/2| \ge \delta) \le 2\exp\left(-2(n-k)\left(\frac{\delta}{t} - \frac{k}{2(n-k)}\right)^2\right). $$
\end{lemma}

\begin{proof}
  Let $S$ be the set of $n$ points. Let the fixed points be $y_1, \dots, y_k$ and the random points be $X_1, \dots, X_{n-k}$. Let $m = n-k$.

  First consider the upper tail.
  In the worst case, all $k$ fixed points are at least $t/2 + \delta$.
  If $M_{n,k} \ge t/2 + \delta$, then at least $\lceil (n+1)/2 \rceil$ points in $S$ must be $\ge t/2 + \delta$,
  so at least $\lceil (n+1)/2 \rceil - k$ random points must be $\ge t/2 + \delta$.
  Since $\lceil (n+1)/2 \rceil \ge (n+1)/2$,
  \begin{align}
  	\frac{1}{m}\sum_{i=1}^{m} \indicator{\{X_i \ge t/2 + \delta\}}
	\ge \frac{\lceil (n+1)/2 \rceil - k}{m}
	\ge \frac{(n+1)/2 - k}{m}
	= \frac{1}{2} - \frac{k-1}{2m}
	\ge \frac{1}{2} - \frac{k}{2m}.
  \end{align}
  Since $\E\left[\frac{1}{m}\sum_{i=1}^{m} \indicator{\{X_i \ge t/2 +
  \delta\}}\right] = \frac{1}{2} - \frac{\delta}{t}$,
  applying Hoeffding's inequality with deviation $\epsilon = \frac{\delta}{t} - \frac{k}{2m}$ gives
  \begin{align}
      P(M_{n,k} \ge t/2 + \delta)
      &\le P\left(\frac{1}{m}\sum_{i=1}^{m} \indicator{\{X_i \ge t/2 + \delta\}} \ge \frac{1}{2} - \frac{k}{2m} \right) \\
      &\le \exp\left(-2m\left(\frac{\delta}{t} - \frac{k}{2m}\right)^2\right).
  \end{align}

  Now consider the lower tail.
  In the worst case, all $k$ fixed points are at most $t/2 - \delta$.
  If $M_{n,k} \le t/2 - \delta$, then at least $\lceil n/2 \rceil$ points in $S$ must be $\le t/2 - \delta$,
  so at least $\lceil n/2 \rceil - k$ random points must be $\le t/2 - \delta$.
  Since $\lceil n/2 \rceil \ge n/2$,
  \begin{align}
  	\frac{1}{m}\sum_{i=1}^{m} \indicator{\{X_i \le t/2 - \delta\}}
	\ge \frac{\lceil n/2 \rceil - k}{m}
	\ge \frac{n/2 - k}{m}
	= \frac{1}{2} - \frac{k}{2m}.
  \end{align}
  Since $\E\left[\frac{1}{m}\sum_{i=1}^{m} \indicator{\{X_i \le t/2 -
  \delta\}}\right] = \frac{1}{2} - \frac{\delta}{t}$,
  applying Hoeffding's inequality with deviation $\epsilon = \frac{\delta}{t} - \frac{k}{2m}$ gives
  \begin{align}
      P(M_{n,k} \le t/2 - \delta)
      &\le P\left(\frac{1}{m}\sum_{i=1}^{m} \indicator{\{X_i \le t/2 - \delta\}} \ge \frac{1}{2} - \frac{k}{2m} \right) \\
      &\le \exp\left(-2m\left(\frac{\delta}{t} - \frac{k}{2m}\right)^2\right).
  \end{align}

  Combining both tails completes the proof.
\end{proof}

\bibliographystyle{plain}
\bibliography{refs.bib}

\begin{thebibliography}{10}

\bibitem{AryaMNSW1998}
Sunil Arya, David~M. Mount, Nathan~S. Netanyahu, Ruth Silverman, and Angela~Y. Wu.
\newblock An optimal algorithm for approximate nearest neighbor searching fixed dimensions.
\newblock {\em Journal of the ACM}, 45(6):891--923, 1998.

\bibitem{Bentley1975}
Jon~Louis Bentley.
\newblock Multidimensional binary search trees used for associative searching.
\newblock {\em Communications of the ACM}, 18(9):509--517, 1975.

\bibitem{Birnbaum1942}
Zygmunt~W. Birnbaum.
\newblock An inequality for {M}ill's ratio.
\newblock {\em The Annals of Mathematical Statistics}, 13(2):245--246, 1942.

\bibitem{ChenFangShao2013}
Louis H.~Y. Chen, Xiao Fang, and Qi-Man Shao.
\newblock From stein identities to moderate deviations.
\newblock {\em The Annals of Probability}, 41(1):262--293, 2013.

\bibitem{DasguptaF2008}
Sanjoy Dasgupta and Yoav Freund.
\newblock Random projection trees and low dimensional manifolds.
\newblock In {\em Proceedings of the 40th Annual ACM Symposium on Theory of Computing (STOC)}, pages 537--546, 2008.

\bibitem{FriedmanBF1977}
Jerome~H. Friedman, Jon~Louis Bentley, and Raphael~A. Finkel.
\newblock An algorithm for finding best matches in logarithmic expected time.
\newblock {\em ACM Transactions on Mathematical Software}, 3(3):209--226, 1977.

\bibitem{MujaL2014}
Marius Muja and David~G. Lowe.
\newblock Scalable nearest neighbor algorithms for high dimensional data.
\newblock {\em IEEE Transactions on Pattern Analysis and Machine Intelligence}, 36(11):2227--2240, 2014.

\bibitem{NeneN1997}
Sameer~A. Nene and Shree~K. Nayar.
\newblock A simple algorithm for nearest neighbor search in high dimensions.
\newblock {\em IEEE Transactions on Pattern Analysis and Machine Intelligence}, 19(9):989--1003, 1997.

\bibitem{Panigrahy2008}
Rina Panigrahy.
\newblock An improved algorithm finding nearest neighbor using kd-trees.
\newblock In {\em LATIN 2008: Theoretical Informatics}, volume 4957 of {\em LNCS}, pages 387--398, 2008.

\bibitem{RamS2019}
Parikshit Ram and Kaushik Sinha.
\newblock Revisiting kd-tree for nearest neighbor search.
\newblock In {\em Proceedings of the 25th ACM SIGKDD Conference on Knowledge Discovery and Data Mining}, pages 1378--1388, 2019.

\bibitem{sklearnKDTree}
{scikit-learn developers}.
\newblock {\tt sklearn.neighbors.KDTree} documentation.
\newblock \url{https://scikit-learn.org/stable/modules/generated/sklearn.neighbors.KDTree.html}.
\newblock Accessed 2026-04-30.

\bibitem{Sproull1991}
Robert~F. Sproull.
\newblock Refinements to nearest-neighbor searching in $k$-dimensional trees.
\newblock {\em Algorithmica}, 6:579--589, 1991.

\bibitem{Vempala2012}
Santosh~S. Vempala.
\newblock Randomly-oriented $k$-d trees adapt to intrinsic dimension.
\newblock In {\em Proceedings of the 32nd International Conference on Foundations of Software Technology and Theoretical Computer Science (FSTTCS)}, pages 48--57, 2012.

\bibitem{WeberSB1998}
Roger Weber, Hans-J\"org Schek, and Stephen Blott.
\newblock A quantitative analysis and performance study for similarity-search methods in high-dimensional spaces.
\newblock In {\em Proceedings of the 24th International Conference on Very Large Data Bases (VLDB)}, pages 194--205, 1998.

\end{thebibliography}
\end{document}